\documentclass[%
 reprint,
 amsmath,amssymb,
 aps,
]{revtex4-2}

\usepackage{amsmath}
\usepackage{amssymb}
\usepackage{bm}
\usepackage{nicematrix}
\usepackage{graphicx}
\usepackage[colorlinks=true, allcolors=blue]{hyperref}
\usepackage{amsthm}
\usepackage{float}
\usepackage{mathtools}
\usepackage{multirow}
\usepackage{booktabs}
\usepackage{algorithm}
\usepackage{algpseudocode}
\usepackage{appendix}
\usepackage{xcolor}

\newtheorem{theorem}{Theorem}[section]
\newtheorem{lemma}[theorem]{Lemma}
\newtheorem{proposition}[theorem]{Proposition}

\newtheorem{example}{Example}[section]
\newtheorem{definition}[theorem]{Definition}

\newcommand{\grobner}{Gröbner }

\newcommand{\hl}[1]{\textcolor{red}{#1}}

\newcommand{\faintcline}[1]{\arrayrulecolor{gray!50}\cline{#1}\arrayrulecolor{black}}

\begin{document}
\title{Materialised symmetries of 2D translationally invariant codes}
\author{Andrew Li}
\author{Stephen D. Bartlett}
\author{Andrew C. Doherty}
\author{Campbell K. McLauchlan}
\affiliation{School of Physics, University of Sydney, Sydney, NSW 2006, Australia}
\date{October 2026}

\begin{abstract}
    There has been significant recent interest in near-term qLDPC codes as high-performance alternatives to surface and color codes. One such class of codes is 2-dimensional translationally invariant (TI) codes, such as bivariate bicycle codes, which share similar properties to topological codes. Fundamental objects in the study of such codes are the materialised symmetries, which can be used for the construction of matching-based decoders. These decoders generalise the minimum-weight perfect matching decoder for toric codes and provide similar asymptotic performance guarantees that heuristic decoders such as BP and its variants lack. Despite this, the mathematical structure of symmetries along with their properties under translation and restriction to finite-sized lattices has not been well-studied.
    We describe a decomposition of symmetry spaces of 2D CSS TI codes on infinite lattices into translation-invariant subspaces which allow us to write an explicit basis of symmetries in a plane-wave-like form.
    We then describe how to adapt this infinite lattice basis so that any corresponding rectangular periodic lattice basis can be obtained by restriction to basis elements compatible with the periodic boundaries. This allows us to describe how the symmetry space varies with different rectangular dimensions. We illustrate via examples that it is often straightforward in practice to determine a basis of symmetries on twisted boundaries as well, without needing to recompute the decomposition.
    We comment on the applications of this framework to matching-based decoders and provide examples of symmetries of such codes as the gross code.
\end{abstract}

\maketitle

\section{Introduction}

Quantum error correction has provided a route to suppressing noise in quantum computers beyond the limits of the hardware~\cite{PhysRevA.52.R2493,Knill_2000, Calderbank_1996,Steane_1996}. Topological codes, particularly those in 2D, have proven a popular choice of error correcting code due to their modest physical demands, especially their sparse connectivity~\cite{Kitaev_2003,Dennis_2002}. Such codes are understood to have tradeoffs between locality and performance~\cite{Bravyi_2010,dai2024localityvsquantumcodes}, and thus recent interest has shifted from traditional nearest-neighbour-connected 2D codes, such as the surface and color codes~\cite{bravyi1998quantumcodeslatticeboundary,Bombin_2006}, to longer-range codes in the form of bivariate bicycle (BB) codes and tile codes~\cite{Bravyi_2024,PhysRevA.88.012311,Panteleev_2021,steffan2025tilecodeshighefficiencyquantum,liang2025planarquantumlowdensityparitycheck,wang2026algebrabivariatebicyclesurfacecodes}. These families offer codes implementable in near-term quantum computers with improved parameters and provide an intermediary between traditional 2D codes and asymptotically-good qLDPC codes. Although not local in a nearest-neighbour sense, these codes can be thought of as 2D topological when placed in a family of growing lattices~\cite{Haah_2013,Liang_2025,Chen_2025} and are in addition translationally-invariant (TI). Leveraging topological and TI algebraic machinery to study these codes has led to fault-tolerant gate constructions~\cite{eberhardt2024logicaloperatorsfoldtransversalgates,liang2026selfdualbivariatebicyclecodes,cross2025improvedqldpcsurgerylogical}, to architecture designs for platforms such as neutral atoms~\cite{Viszlai_2026,Wang_2026} and superconducting qubits~\cite{yoder2025tourgrossmodularquantum,Shaw_2025,nixon2026vinecodeslowoverheadquantum,gu2026nearestneighbourgatesneedhighrate}, and to decoders~\cite{sahay2026matchingdecoderbivariatebicycle,tan2026generalizedmatchingdecoders2d}.

Studies of the structure of topological codes has motivated interest in \emph{materialised symmetries}~\cite{Kitaev_2003,Brown2020a} and the closely related concept of \emph{metachecks}~\cite{Campbell_2019}, where groups of stabilizer generators multiply to give the identity. Materialized symmetries lead to matchable decoding subproblems, whose solutions are combined to form a decoding~\cite{Brown2020a,Delfosse_2014,Brown_2022,Sahay_2022,Kubica_2023}. These decoders are of interest as they have provable thresholds and thus can have arbitrarily high error suppression on sufficiently large lattices \cite{Raveendran_2021,tan2026generalizedmatchingdecoders2d}. This is in contrast to heuristic decoders such as belief propagation (BP) which may hit error floors at low physical error rates. Materialized symmetries are also relevant in decoupling 2D topological phases \cite{Kubica_2015,wang2026decoupling2dtranslationinvarianttopological} and designing planar boundaries \cite{liang2026operatoralgebraalgorithmicconstruction,liang2025planarquantumlowdensityparitycheck}. 

Motivated by the rich translational structure of materialized symmetries in quantum error correcting codes of interest, in this paper we explore this structure through the lens of Haah's formalism for describing 2D CSS TI codes \cite{Haah_2013} on an infinite lattice. In this formalism, the checks of the code are described by a group algebra over the lattice, and translation corresponds to multiplication. We first highlight some key translational properties of symmetries, particularly their components closed under translation, with a motivating example in Section~\ref{sec-motivating} and relate it to the concept of anyon-symmetry duality. We cover necessary background in Section~\ref{sec-background}, in particular formally describing anyon-symmetry duality, and then in Section~\ref{section-infinite} describe a decomposition of the space of symmetries into subspaces closed under translation by leveraging anyon-symmetry duality, and obtain an explicit basis for each subspace. Each subspace corresponds to coordinates in some extension field of $\mathbb{F}_2$, and these coordinates describe the horizontal/vertical frequencies of their corresponding symmetries, which in the simplest cases resemble plane waves. In Section~\ref{section-finite}, we describe how to use infinite-lattice symmetries to compute a basis of symmetries after compactifying to any periodic rectangular-lattice configuration. We find that after appropriately adapting the infinite-lattice basis for each subspace, we can efficiently determine finite-lattice symmetries and explicitly describe how the symmetry space varies with boundary dimensions. 
We also discuss how to use these methods to analyse the case of twisted boundary conditions for less complex decompositions.

\section{Motivating example: 2D Color Code}\label{sec-motivating}

Let us first consider the motivating example of the 2D hexagonal color code, which we will set on an infinite lattice. We pick this example for the relative simplicity of its materialized symmetries, which nevertheless display interesting patterns on the 2D lattice that are not fully described by standard approaches such as Gaussian elimination. These materialized symmetries and related metachecks are well-known for their application in matching-based decoding~\cite{Delfosse_2014,Sahay_2022,Kubica_2023}. We will also informally introduce some preliminaries relevant to the example such as anyon-symmetry duality, starting with the anyons of the color code and seeing how we can construct metachecks from these.

\begin{figure*}[htbp]\
    \centering
    \includegraphics[width=0.32\textwidth]{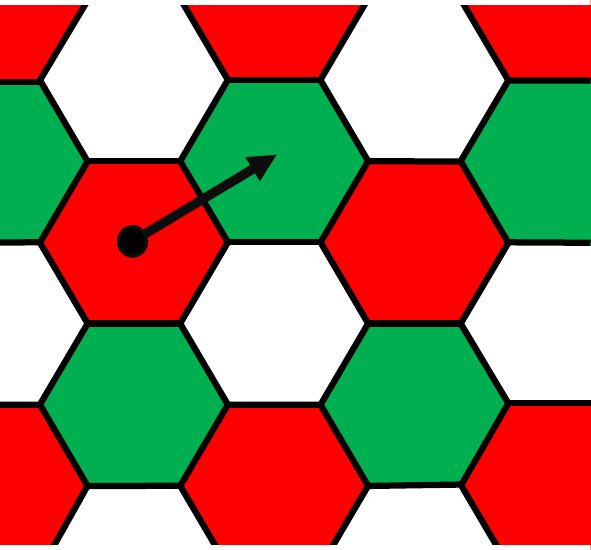}
    \hfill
    \includegraphics[width=0.32\textwidth]{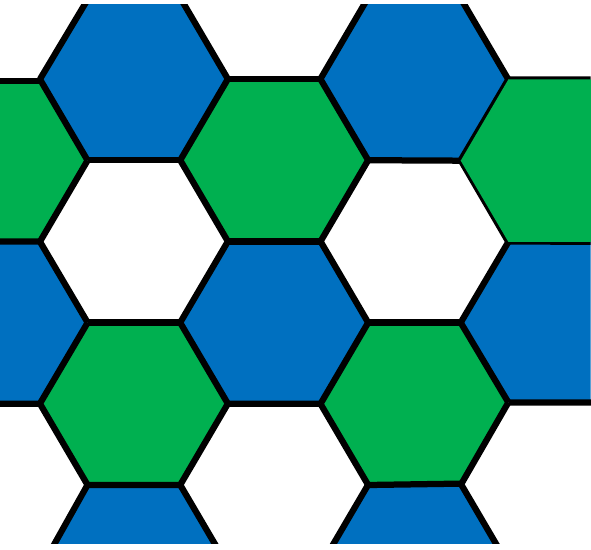}
    \hfill
    \includegraphics[width=0.32\textwidth]{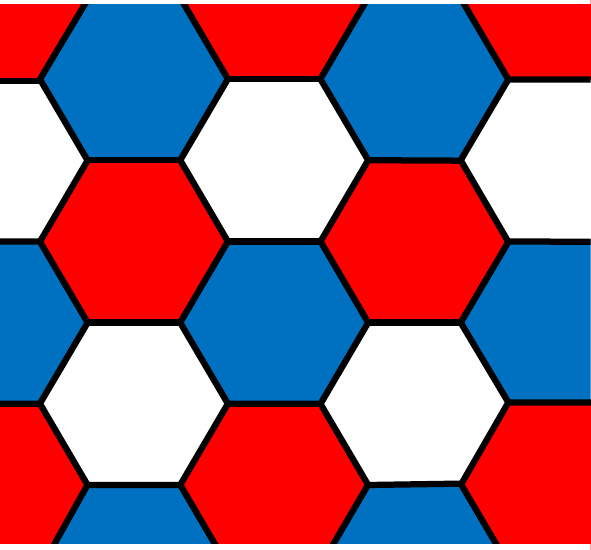}
    \caption{The \textit{rg}, \textit{gb}, and \textit{br} X or Z symmetries of the hexagonal color code. The arrow is a translation that cycles between symmetries.}
    \label{fig:cc_symmetries}
\end{figure*}

For our purposes, anyons are a classification of syndromes such that two syndromes describe the same anyon if we can deform one into the other via some finite-weight Pauli operator. In the color code, each violated $\mathbf{c}$-colored X check (with $\mathbf{c}\in \{\mathbf{r},\mathbf{g},\mathbf{b}\}$) hosts an anyon, named $\mathbf{cz}$. Analogously, violated Z checks host $\mathbf{cx}$ anyons. We can also fuse anyons; two anyons $a$ and $b$ fuse to give $c$ if some application of finite-weight Pauli operators converts the former pair into the latter. For example, a single Pauli X error converts a pair of adjacent $\mathbf{rx}$ and $\mathbf{gx}$ anyons into a $\mathbf{bx}$ anyon; indeed $\mathbf{c_1x} \times \mathbf{c_2x} = \mathbf{c_3x}$ for distinct $\mathbf{c_1},\mathbf{c_2},\mathbf{c_3}$ (and where $\times$ represents fusion). A materialised symmetry is an assignment of binary values to anyons that respects the fusion structure. For example, if we assign $\mathbf{rx} \to 1$ and $\mathbf{gx} \to 1$, then we must assign $1+1=0$ to $\mathbf{bx} = \mathbf{rx} \times \mathbf{gx}$. Symmetries describe conservation rules; the total parity of all anyons assigned 1 by a symmetry, within a valid syndrome resulting from finite-weight errors, is even. The above example, which we call the \textit{rg} Z symmetry (as the violated checks have Pauli label Z), says that the total number of $\mathbf{rx}$ and $\mathbf{gx}$ anyons is always 0 (mod 2). We analogously have \textit{gb} and $\textit{br}$ Z symmetries. We can add symmetries by adding the corresponding assignments pointwise (mod 2); for example the \textit{rg} and \textit{gb} Z symmetries add to give to the \textit{br} Z symmetry. 

We can now use these symmetries to construct metachecks, patterns of stabiliser checks that multiply to the identity; these tell us about redundancies between stabiliser checks. In this work, we use the term (materialised) symmetry to describe the objects dual to anyons, and metachecks to refer to subsets of checks that multiply to the identity. Looking at the \textit{rg} Z symmetry, for example, we know that any error must create $\mathbf{rx}$ and $\mathbf{gx}$ anyons in pairs. Consider a pattern consisting of all stabilisers whose violations are $\mathbf{rx}$/$\mathbf{gx}$ anyons, i.e. the red and green Z checks. Any X error will commute with this pattern due to its corresponding conservation rule (explicitly, it violates an even number of checks in the pattern), and thus the pattern must be a metacheck. Indeed, every symmetry gives a unique metacheck and vice versa, so we regard the two as equivalent. We show the metachecks constructed from the \textit{rg}, \textit{gb}, and \textit{br} Z symmetries in Fig.~\ref{fig:cc_symmetries}.

A standard way to construct metachecks is as follows: viewing the group of anyons as a vector space with addition corresponding to fusion, we compute a basis of anyons (e.g. $\{\mathbf{rx},\mathbf{gx}\}$). We can then compute the dual basis of the symmetries: for each basis anyon $a$, define the dual basis element to assign $a$ to 1 and all other basis anyons to 0. We then acquire a basis of metachecks by including exactly the checks whose violations correspond to anyons assigned 1 (i.e., red and green $Z$ plaquettes for symmetries dual to $\mathbf{rx}$ and $\mathbf{gx}$ anyons).

What this method does not generally provide, however, is insight into the translational structure of these metachecks. For instance, we can observe that there exists a translation which cycles between the three metachecks of each Pauli label (Fig.~\ref{fig:cc_symmetries}). This suggests a translational periodicity of 3 and thus tells us that when introducing periodic boundary conditions, we need these boundary lengths to be multiples of 3 in order to support these metachecks. As the dimension of the space of metachecks is proportional to the logical dimension; it is very much necessary that the color code is placed on such suitable boundary conditions. We would thus like to explore the symmetries of 2D codes, not just considering the pointwise addition but also the translational data via a multiplication operation. We will decompose the symmetry space into `components'; each component is a subspace of symmetries closed under translation such that the entire symmetry space is a direct sum of these components. Loosely speaking, we can think of these components as the frequencies/periodicities that make up an arbitrary symmetry/metacheck. In the color code, there is only one component of X (or Z) symmetries.

\section{Background}\label{sec-background}

In this section, we briefly describe the chain complex formulation of CSS codes, and the additional module structure which the CSS chain complex is equipped with when considering TI codes on a square lattice. We then introduce 2D TI codes in this framework, along with key concepts such as anyons and metachecks/symmetries, which we show to be dual to each other.

\subsection{CSS codes as chain complexes}

A chain complex is a sequence $\{C_n\}$ of objects and maps $\partial_n : C_n \to C_{n-1}$ called differentials such that $\partial_n \partial_{n+1} = 0$. Typically, each $C_n$ will be a $\mathbb{F}_2$-vector space and the differentials will be $\mathbb{F}_2$-linear maps. Of interest are length-2 complexes \cite{bravyi2013homologicalproductcodes}:
\begin{align}
    C_2 \xrightarrow[]{\partial_2} C_1 \xrightarrow[]{\partial_1} C_0
\end{align}
which describe CSS codes~\cite{Calderbank_1996,Steane_1996} after selecting bases of $C_2$, $C_1$, and $C_0$ to represent the X checks, qubits, and Z checks respectively. Here, $\partial_2$ takes an X stabiliser to its support in qubits and is represented by the transposed X parity check matrix. Also, $\partial_1$ takes the support of an X operator to its Z syndrome, and is represented by the Z parity check matrix. 
In $C_1$, the X logical operators have no Z syndrome and are thus given as $\ker \partial_1$. When quotiented by the X stabilisers (as their qubit supports) $\mathrm{im}\, \partial_2$ we get the homology $H_1$ at $C_1$ representing the X logical classes.
We can also construct the dual complex:
\begin{align}
    C^0 \xrightarrow[]{d^0} C^1 \xrightarrow[]{d^1} C^2
\end{align}
where $d^0 = \partial_1^T$ and $d^1 = \partial_2^T$. This is the Z/X flipped counterpart of the original CSS chain complex and takes Z stabilisers to qubits, and qubits to their X syndromes.

\subsection{2D translationally invariant codes}

We briefly review CSS codes on an infinite 2D square lattice, using Haah's Laurent polynomial formalism \cite{Haah_2013}. By replacing periodic boundaries with the infinite lattice, we can explore properties that are shared between families of codes across different lattices, particularly anyons and symmetries.

The sites of this lattice are labelled by the group $\mathbb{Z}^2$; we write the site at co-ordinates $(i,j)$ as $x^i y^j$; as a group action, these elements describe translations. If we use the elements of $\mathbb{Z}^2$ to label our checks or qubits, then they form a basis spanning the Laurent polynomial ring $R:=\mathbb{F}_2[x^{\pm1},y^{\pm1}]$. This ring allows us to describe any finite set of co-ordinates in $\mathbb{Z}^2$.
We will consider a setup with two qubits on each site, such that $C_1 = R^2$. Although one can construct examples with more qubits per site, the vast majority of 2D TI codes of interest (e.g.,\ BB codes) fall under this construction. We will demand that all translations of a stabiliser are themselves stabilisers, a condition called translational-invariance. 
Suppose there is one type of X check, and all X checks are translations of this one type. We may choose to label one such check as $1$, and then its $(i,j)$-translations as $x^i y^j$. Under this labelling $C_2 = R$, and $\partial_2$ is not only $\mathbb{F}_2$-linear but $R$-linear, that is $\partial_2(x^i y^j s) = x^i y^j \partial_2(s)$. Similarly, setting one type of Z check we get that $C_0 = R$ and $\partial_1$ is $R$-linear. We then have that the CSS chain complex is a complex of $R$-modules.

We will additionally demand that the CSS complex is exact at $C_1$, which is to say that $H_1 = 0$ (this will ensure that we have a growing distance when compactifying to finite lattices). This then forces the complex to take the following form:
\begin{align}
    R \xrightarrow[]{\left[\begin{smallmatrix}
        A \\ B
    \end{smallmatrix}\right]}
    R^2 \xrightarrow[]{\left[\begin{smallmatrix}
        B & A
    \end{smallmatrix}\right]}
    R 
\end{align}
where $A, B \in R$ are Laurent polynomials satisfying $\mathrm{gcd}(A,B) = 1$~\cite{Liang_2025,wang2026algebrabivariatebicyclesurfacecodes}. Swapping X and Z, we get:
\begin{align}
    R \xrightarrow[]{\left[\begin{smallmatrix}
        \bar{B} \\ \bar{A}
    \end{smallmatrix}\right]}
    R^2 \xrightarrow[]{\left[\begin{smallmatrix}
        \bar{A} & \bar{B}
    \end{smallmatrix}\right]}
    R 
\end{align}
where $\overline{x^i y^j} = x^{-i} y^{-j}$.
We can `compactify' this code to periodic boundaries by replacing $R$ with $R/\mathfrak{b}$ where $\mathfrak{b} = (x^\ell - 1, y^m - 1)$ describes a rectangular torus with dimensions $\ell \times m$. More generally, we might have $\mathfrak{b} = (x^\ell - 1, x^my^n - 1)$ describing a twisted torus occupying a parallelogram with corners $(0,0)$, $(\ell,0)$,$(m,n)$, and $(m+\ell,n)$.

\subsection{Anyons and symmetries}

We will now review the dual notions of anyons and symmetries in the context of Haah's formalism. 
In our infinite lattice setting, an anyon refers to any syndrome pattern that can be produced by an infinite-support Pauli operator~\cite{Haah_2013}; for all codes we consider, this is simply all syndrome patterns. Let us start by considering Z check syndrome patterns, which we call $m$-anyons. We can classify these $m$-anyons into `sectors'; two $m$-anyons share a sector if we can turn one into the other by some finite-support X operator. We can obtain the space of $m$-anyon sectors by quotienting out the space of Z syndromes $C_0 = R$ by all possible syndromes created by a finite-support X operator, $\mathrm{im}\, \partial_1  = (A,B)$ to get $R/(A,B)$~\cite{Haah_2013, Liang_2025}. Analogously, X syndrome patterns or $e$-anyons are classified by $R/(\overline{A},\overline{B})$.

A metacheck is a pattern of stabiliser checks which multiply to the identity. We call such a pattern of Z checks an $m$-metacheck and its X counterpart an $e$-metacheck.
We can describe the space of finite-support or local $m$-metachecks with $\ker d^0$ (in our construction, $d^0$ is injective so there are no local metachecks), however as $R$ only supports finite collections of lattice sites, we need to extend it to consider infinite-support or global metachecks. To this end, we will need to describe the dual $M^*$ of an $R$-module $M$.

\begin{definition}
    Given an $R$-module $M$ (with $R$ a Laurent-polynomial ring), the dual of $M$, denoted $M^*$, is the $\mathbb{F}_2$-vector space dual of $M$ equipped with the $R$-multiplication operation $(r \phi) (m) = \phi (\bar{r} m)$ for $r \in R$, $\phi \in M^*$, and $m \in M$.
\end{definition}

In particular, identifying $x$ and $y$ with their $\mathbb{F}_2$-vector space duals (the dual of $x$ being the function taking $x$ to 1 and all other monomials to 0), we can think of $R^*$ as containing $R$; while $R$ only permits finite sums of monomials, $R^*$ allows any infinite series of monomials as we can map as many monomials to $1$ as we like.
By replacing $R$ with $R^*$ in the CSS chain complex, we can then write the space of global $m$-metachecks as $\ker \widetilde{d^0}$ where $\widetilde{d^0} = \left[\begin{smallmatrix}
        \bar{B} \\ \bar{A}
    \end{smallmatrix}\right]$
is represented by the same matrix as $d^0$ but acts on $R^*$ instead of $R$. Analogously, we can describe the space of global $e$-metachecks as $\ker \tilde{\partial_2}$.
We now state a result connecting metachecks to anyons:

\begin{proposition}[Anyon/metacheck duality]\label{prop-am-duality}
    The $e$/$m$-anyons are dual to the $e$/$m$-metachecks, that is:
    \begin{align*}
        \ker \widetilde{d^0} &\cong (R/(A,B))^*
        \\
        \ker \widetilde{\partial_2} &\cong (R/(\overline{A},\overline{B}))^*
    \end{align*}
\end{proposition}
\begin{proof}
    See Appendix~\ref{proof1}
\end{proof}

As discussed in the motivating example, this picture allows us to think of metachecks as anyon conservation rules, a context in which they are often referred to alternatively as symmetries. We can also think of metachecks as stabiliser patterns following a consistency rule enforced by the anyon structure; if a metacheck has odd overlap with some syndrome pattern, then it must have odd overlap with all syndrome patterns belonging to the same anyon sector and vice versa. As the anyon structure informs the metacheck structure, once we have computed the former the latter follows; a straightforward way to obtain a basis of metachecks is to first compute a basis of anyons and then take the dual. The anyon basis, as demonstrated in recent works~\cite{Liang_2025}, can be readily computed by taking the monomials irreducible by a \grobner basis of $(A,B)$, although such a basis does not generally provide insight into the translational structure of metachecks. We therefore take a different approach, outlined below.

\section{Symmetries on the infinite lattice}\label{section-infinite}

\subsection{Overview of the procedure}

We will in this section consider symmetries on the infinite lattice. Although these symmetries have infinite support, we will see that upon decomposing the symmetry space into components closed under translation, each component will become `local' under a change of basis and thus admit a compact description. 

First, we solve the system of equations $A=B=0$ as its solutions label the components of the decomposition. The solutions may fall in field extensions of $\mathbb{F}_2$, but there is some sufficiently large finite field $\mathbb{K}$ that contains all such solutions. Upon extending to this finite field to get what we informally call $\mathbb{K}$-anyons (as opposed to the original $\mathbb{F}_2$-anyons), we can decompose our anyon space into a direct sum of components, each of which is labelled by some point $(x=\alpha, y=\beta)$ that is a common root of $A$ and $B$. 
These points organise into orbits of Galois conjugates, analogous to complex conjugate pairs when solving systems of real polynomials. We then take the dual of each component, to get a similar decomposition of the $\mathbb{K}$-symmetry space. We then apply a map called the field trace to return to $\mathbb{F}_2$-symmetries, upon which components within an orbit collapse to a single component. We can thus also associate each component on the $\mathbb{F}_2$-symmetry space with an orbit; the component will describe a subspace closed under translation that is irreducible in the sense that it cannot be further decomposed. The points composing the associated orbit as well as the local structure of the component at these points will allow us to determine properties of the subspace such as translational periodicity.

\subsection{Decomposing the symmetry space}

We will focus on $m$-symmetries, which we will for convenience simply refer to as symmetries or metachecks. The procedure for computing $e$-symmetries follows analogously.

As we have seen, the $m$-symmetries are associated with elements of the ring $R/(A,B)$. We will from here onwards define:
\begin{align}
    I = (A,B) \subseteq R
\end{align}
so that the symmetries are $R/I$.
In this ring, polynomials $A$ and $B$ are identified with zero. Consider, therefore, the system of equations $A(x,y) = B(x,y) = 0$. There exists some finite field $\mathbb{K} \supseteq \mathbb{F}_2$ in which all solutions lie.
As we are working in the Laurent polynomial ring where $x$ and $y$ are invertible, we ignore solutions where $x=0$ or $y=0$. Let us denote the set of remaining solutions by $\mathcal{Z}$. Under our assumption that the CSS complex is exact at $C_1$, there will be finitely many such solutions.
After extending $R$ by $\mathbb{K}$ to get $R' = \mathbb{K}[x,y,x^{-1},y^{-1}] \supseteq R$, our anyon space extends to $R'/J$, where $J=IR' := \lbrace xr \, | \, x\in I, r\in R'\rbrace$ is the extension ideal of $I$.
We informally call this space the $\mathbb{K}$-anyons. The ideal $J$ admits a primary decomposition:
\begin{align}
    J = \bigcap_{(\alpha,\beta) \in \mathcal{Z}}\;J_{(\alpha,\beta)}
\end{align}
with a unique ideal $J_{(\alpha,\beta)}$ associated to each point of $\mathcal{Z}$ such that $\sqrt{J_{(\alpha,\beta)}} = (x-\alpha,y-\beta)$ (we define $\sqrt{Q} := \{r \in R \, | \, r^n \in Q \;\mathrm{for} \;\mathrm{some}\; n \in \mathbb{Z}_{\geq 0}\}$). This leads to the following decomposition of the $\mathbb{K}$-anyon space:

\begin{figure}
    \centering
    \includegraphics[width=1.0\linewidth]{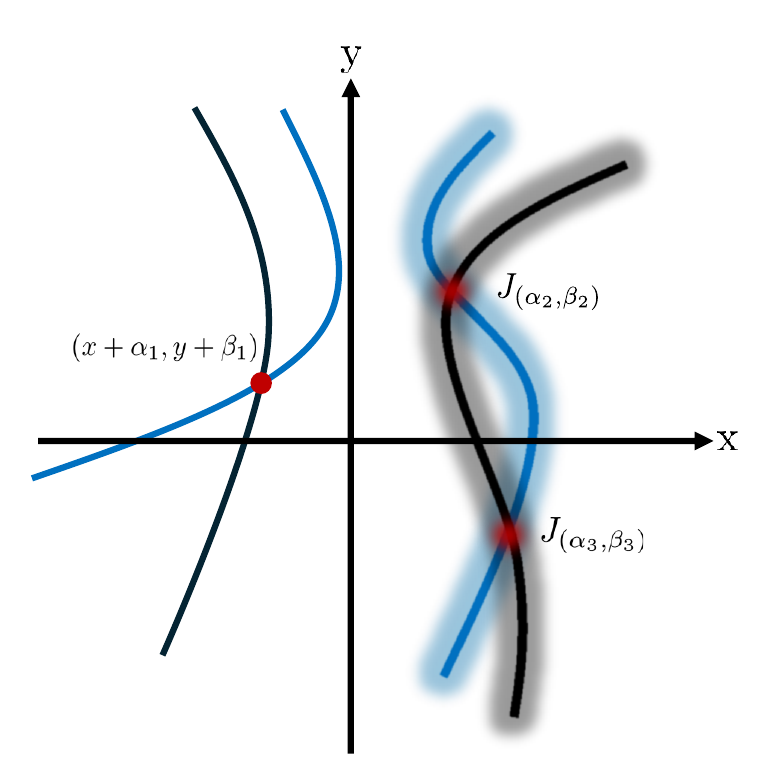}
    \caption{Schematic illustration of the decomposition of $R'/J$. Each component corresponds to polynomial functions at a point $(\alpha_i,\beta_i)$, but fat points ($i=2,3$) include additional structure in their neighbourhood. The points occur at the intersection of $A=0$ (black) and $B=0$ (blue).
    }
    \label{fig:primary-decomp}
\end{figure}

\begin{proposition}[Chinese remainder theorem decomposition \cite{Eisenbud1995}]\label{crt-decomp}
\begin{align}
    R'/J \cong \bigoplus_{(\alpha,\beta) \in \mathcal{Z}} \,R'/J_{(\alpha,\beta)}
\end{align}
\end{proposition}

Computing such a decomposition is a standard task in computational commutative algebra. Importantly, each component $R'/J_{(\alpha,\beta)}$ of $R'/J$ is associated with some point in $\mathcal{Z}$, and can be thought of as describing the local neighbourhood of $R'/J$ at this point. When $J_{(\alpha,\beta)} \neq (x-\alpha,y-\beta)$, we say that $(\alpha,\beta)$ is a fat point (Fig.~\ref{fig:primary-decomp}), corresponding to a root with non-trivial multiplicity.
We will from now onwards, when describing $R'/J_{(\alpha,\beta)}$ centre our coordinates at the relevant point and write its elements in terms of $u=x-\alpha$ and $v=y-\beta$ instead of $x$ and $y$. Note that each component is closed under translation by virtue of being an $R'$-module. Something that will become relevant later when returning to $\mathbb{F}_2$ is that the points of $\mathcal{Z}$ organise themselves in orbits of the form $\{(\alpha,\beta),(\alpha^2,\beta^2),(\alpha^4,\beta^4),\dots\}$. These are the finite field analogue of complex conjugate pairs from solving systems of real polynomials.
\begin{definition}[Galois orbits of roots]
    An orbit of $\mathcal{Z}$ is a length-$n$ sequence of solutions in $\mathcal{Z}$ taking the form $(\alpha_i,\beta_i)_{i \in \mathbb{Z}_n}$ where $\alpha_{i+1} = \alpha_i^2$ and $\beta_{i+1} = \beta_i^2$. $\mathcal{Z}$ can be partitioned into finitely many orbits.
\end{definition}

The smallest field in which each orbit $\mathcal{O}$ can be embedded is $\mathbb{F}_{2^{|\mathcal{O}|}}$. We can now make the move from $\mathbb{K}$-anyons to $\mathbb{K}$-symmetries by taking the dual (the space of $\mathbb{K}$-symmetries is then $(R'/J)^*$, where $^*$ denotes the $\mathbb{K}$-vector space rather than $\mathbb{F}_2$-vector space dual), and doing so preserves the decomposition by taking $\mathbb{K}$-anyon components to $\mathbb{K}$-symmetry components.

\begin{proposition}\label{prop-dual-decomp}
    $(R'/J)^* \cong \oplus_{(\alpha,\beta) \in \mathcal{Z}} \left(R'/J_{(\alpha,\beta)}\right)^*$
\end{proposition} 
\begin{proof}
    See Appendix~\ref{proof2}
\end{proof}

If we have computed a basis of each $\mathbb{K}$-anyon component, then its dual is a basis of the corresponding $\mathbb{K}$-symmetry component. We will now look to return to the space of $\mathbb{F}_2$-symmetries via a $\mathbb{F}_2$-linear map that takes an $\mathbb{F}_2$-vector space basis ($\mathbb{F}_2$-basis) of $\mathbb{K}$-symmetries onto a basis of $\mathbb{F}_2$-symmetries in a way that takes components to components. We will thus need to first turn our $\mathbb{K}$-vector space basis ($\mathbb{K}$-basis) into a $\mathbb{F}_2$-basis of the same space. This is achieved by taking the setwise product of a $\mathbb{F}_2$-basis $\{a_i\}$ with the $\mathbb{K}$-basis $\{b_i\}$ of our $\mathbb{K}$-symmetries to get $\{a_i b_j\}$.
Before we describe this map, let us first define the field trace:

\begin{definition}
    Given a finite field $\mathbb{K} \supseteq \mathbb{F}_2$ and some $\alpha \in \mathbb{K}$, the field trace $\mathrm{Tr}(\alpha)$ is a $\mathbb{F}_2$-linear map given by the sum over an orbit:
    \begin{align}
        \mathrm{Tr}(\alpha) = \alpha + \alpha^2 + \alpha^4 + \dots + \alpha^{2^{n-1}}
    \end{align}
    where $\mathbb{K}$ has size $2^n$. We overload notation to define the trace of a $\mathbb{K}$-symmetry as follows:
    \begin{align*}
        (\mathrm{Tr}(\phi))(r) = \mathrm{Tr}(\phi(r))
    \end{align*}
    which is to say that we take the trace pointwise.
\end{definition}

We can now describe the map returning us to $\mathbb{F}_2$-symmetries. Given a symmetry $\phi$, we take its field trace and then restrict it to $R/I$ via the natural embedding of $R/I$ within $R'/J$ (concretely, $r + I$ is embedded as $r + J$). We claim that this map takes $\mathbb{K}$-symmetry components belonging to the same orbit onto the same $\mathbb{F}_2$-symmetries and thus we only need to apply it to one component per orbit.

\begin{theorem}\label{prop-trace-decomp}
    Let the image of a $\mathbb{K}$-symmetry component $(R'/J_{(\alpha,\beta)})^*$ under the $\mathbb{F}_2$-linear map 
    $$s \mapsto \mathrm{Tr}(s) \circ \iota$$ 
    (where $\iota$ is the natural embedding of $R/I$ into $R'/J$) be denoted $S_{(\alpha,\beta)}$. 
    Then $S_{(\alpha,\beta)} = S_{(\alpha',\beta')}$ for $(\alpha,\beta)$ and $(\alpha',\beta')$ belonging to the same orbit and we can label these $S_O$ by orbits $O \in \mathcal{O}$.
    Here, $S_O$ is an $R$-module (it is closed under translation) and 
    \begin{align}
        (R/I)^* \cong \bigoplus_{O \in \mathcal{O}} S_O.
    \end{align}
\end{theorem}

\begin{proof}
    See Appendix~\ref{proof3}.
\end{proof}

This theorem concludes our procedure: compute a decomposition of $\mathbb{K}$-anyons, take duals to get a decomposition of $\mathbb{K}$-symmetries, and then apply the above map to get a decomposition of $\mathbb{F}_2$-symmetries. Practically, we will want to choose the field $\mathbb{K}$ component-by-component, to be the smallest field $\mathbb{F}_{2^{\vert \mathcal{O} \vert}}$ containing the orbit associated with each component. This ensures that when taking the trace-restriction of the $\mathbb{F}_2$-basis of $\mathbb{K}$-symmetries to get a spanning set of $\mathbb{F}_2$-symmetries, we will not get any redundant elements and it will be a basis. To find each orbit in its smallest field, we take the primary decomposition $\{P_i,I_i\}$ of $I$ (where $I_i$ are the primary components and $P_i = \sqrt{I_i}$). Each pair is associated with some orbit, and the smallest field containing that orbit is $\mathbb{F}_{2^{d_i}}$ where $d_i = \mathrm{dim}_{\mathbb{F}_2}(R/P_i)$. Instead of a single $R' = \mathbb{K}[x^{\pm1},y^{\pm1}]$, we will select $R_i' = \mathbb{F}_{2^{d_i}}[x^{\pm1},y^{\pm1}]$ for each orbit, whose $\mathbb{F}_{2^{d_i}}$-symmetries can then be computed by decomposing $(R_i'/J_i)^*$, where $J_i = I_i R_i'$.

Let us now illustrate this procedure with some examples.  

\begin{figure}
    \centering
    \includegraphics[width=1.0\linewidth]{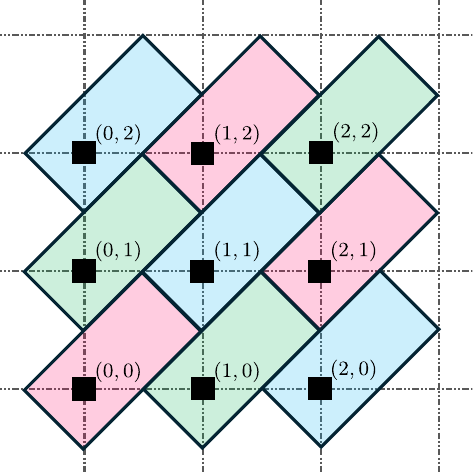}
    \caption{A patch of color code on a planar lattice, as described by the polynomials $A = 1 + x + xy$, $B = 1 + y + xy$. Note how each color of check lives along sets of diagonals $(i,j) \equiv c$ (mod 3) with $c=0$ for red, $c=1$ for green, and $c=2$ for blue.}
    \label{fig:color_code_bb}
\end{figure}

\subsection{Example:  2D hexagonal color code}

We will begin with the motivating color code example of Section~\ref{sec-motivating}, and will show how the observed translational structure is made apparent by this procedure.

By separating the two qubits at each lattice site by a $(1/2,1/2)$ shift, (Fig.~\ref{fig:color_code_bb}), we can implement the hexagonal color code on a square lattice with the polynomials $A = 1 + x + xy$, $B = 1 + y + xy$~\cite{eberhardt2024pruningqldpccodesbivariate}. 
    
We first solve the system of equations $A=B=0$. Taking $A+B = x + y =0$, we have $x=y$. Substituting this back into $A$, we get $x^2 + x + 1 = 0$. This gives solutions $(\omega,\omega)$ and $(\omega^2,\omega^2)$, where $\omega$ is a primitive element of $\mathbb{F}_4$ satisfying $\omega^2 + \omega + 1 = 0$. Note that these two points form a single orbit. Writing $J = (x+y,1+x+x^2) = (x+y,x+\omega) \cap (x+y,x+\omega^2) = (x+\omega,y+\omega) \cap (x+\omega^2,y+\omega^2)$, we see that none of the points are fat. As we only need to select a single orbit representative, we select $(\omega,\omega)$, with corresponding component $R'/(u,v) \cong \mathbb{F}_4$ (with $u=x+\omega$, $v=y+\omega$). Its sole basis element is 1, and thus a basis of the dual component of $\mathbb{F}_4$-symmetries is $\{\delta_1\}$ (being the map taking $1 \mapsto 1$ and all other monomials in $u$ and $v$ to 0). This $\mathbb{F}_4$-metacheck evaluates on $x^i y^j = (u+\omega)^i (v+\omega)^j$ as $\omega^{i+j}$. 

An  $\mathbb{F}_2$-basis of $\mathbb{F}_4$ is $\{1,\omega\}$ and thus we can write a $\mathbb{F}_2$-basis of $\mathbb{F}_4$-symmetries: $\{\delta_1, \omega \delta_1\}$. We then take the field trace of this basis to obtain the following basis of the component $S_{(\omega,\omega)}$ of $\mathbb{F}_2$-symmetries: $\{x^i y^j \mapsto \mathrm{Tr}(\omega^{i+j}), x^i y^j \mapsto \mathrm{Tr}(\omega^{i+j+1})\}$, forming a 2D space with three non-trivial elements $\{x^iy^j \mapsto \omega^{i+j+c} \;\vert\; c \in \mathbb{Z}_3\}$. As there is only one orbit in $\mathcal{Z}$, this is the sole component of the $\mathbb{F}_2$-symmetry space. These three non-trivial elements correspond to our three colors of metacheck, and any unit translation (e.g. $x^i y^j \mapsto x^{i+1} y^j$) will cycle between these. As $$\mathrm{Tr}({\omega^n}) = \begin{cases} 
      0 & n \equiv 0 \;(\mathrm{mod} \; 3) \\
      1 & \mathrm{otherwise}
\end{cases}$$
we can see that each metacheck is supported on two `colors' of diagonal $i+j = c$ (Fig.~\ref{fig:color_code_bb}).

\subsection{Example:  gross code}
\label{example-gross-code}

The gross code~\cite{Bravyi_2024} is described by polynomials $A = 1 + x + y^3/x$ and $B = 1 + y + x^3/y$. The solutions of $A=B=0$ form two orbits in $\mathbb{F}_4$; the first is generated by $(\omega,\omega)$ and the second is generated by $(\omega,\omega^2)$ (we can think of this as a reflection of $(\omega,\omega)$ and indeed reflecting the previous example's choice of color code polynomials in any axis changes its orbit to this). $J$ decomposes into two ideals; $J_{(\omega,\omega)}$ is $(u,v)$ with $u=x+\omega, v=y+\omega$. This is identical to that of the color code and produces the same metachecks. $J_{(\omega,\omega^2)}$ is $(v^3,u+\omega v)$ with $u=x+\omega$, $v=y+\omega^2$. We have a basis $\{1,v,v^2\}$ of the corresponding component $R'/J_{(\omega,\omega^2)}$ of $\mathbb{F}_4$-anyons which gives rise to the basis $\{\delta_1,\delta_v + \omega \delta_u, \delta_{v^2} + \omega^{2} \delta_{u^2}\}$ of the component $(R'/J_{(\omega,\omega^2)})^*$ of $\mathbb{F}_4$-metachecks (the relation $u+\omega v = 0$ couples $\delta_u$ and $\delta_v$, as well as $\delta_{u^2}$ and $\delta_{v^2}$ together). Changing from $u$-$v$ coordinates to $x$-$y$ coordinates, we have that $\delta_1$ becomes $x^i y^j \mapsto \omega^{i+2j}=\omega^{i-j}$. We have that $\delta_v + \omega \delta_u$ becomes $x^i y^j = (u+\omega)^i (v+\omega^2)^j \mapsto j \omega^{i + 2(j-1)} + \omega i \omega^{(i-1)+2j} = \omega^{i-j}(i + \omega j)$, where we take the values of $i$ and $j$ modulo 2 on the right-hand side. Similarly, $\delta_{v^2} + \omega^2 \delta_{u^2}$ becomes $\omega^{i-j}\left(\binom{i}{2} + ij \omega + \binom{j}{2} \omega^2\right)$. Ultimately, we get the following basis of $\mathbb{F}_2$-metachecks:
    \begin{align*}
        x^i y^j &\mapsto \mathrm{Tr}(c \omega^{i+j}) \\
        x^i y^j &\mapsto \mathrm{Tr}(c \omega^{i-j}) \\
        x^i y^j &\mapsto \mathrm{Tr}(c \omega^{i-j}(i + \omega j)) \\
        x^i y^j &\mapsto \mathrm{Tr}\left(c \omega^{i-j}\left(\binom{i}{2} + ij \omega + \binom{j}{2} \omega^2\right)\right)
    \end{align*}
with $c \in \{1,\omega\}$, giving an 8-dimensional space of metachecks on the infinite lattice. The $\mathbb{K}$-symmetries (e.g. $\omega^{i+j}$) supply the shape of the metacheck, whereas the prefactor $c$ describes the linearly independent shifts of this shape. The first four symmetries are the color code symmetries (with $\omega^{i-j}$ being a reflection of $\omega^{i+j}$). The next two basis metachecks have a factor $i + \omega j$ which is periodic under horizontal/vertical translations of 2.  Combined with the $\omega^{i-j}$ factor we see they are periodic under translations $(2l,2m)$ with $l-m \equiv 0 \;(\mathrm{mod} \; 3)$. The next two basis metachecks are periodic under translations $(4l,4m)$ with $l-m \equiv 0 \;(\mathrm{mod} \; 3)$. We can see that the smallest rectangular lattice hosting all metachecks is $12 \times 12$, and the standard $6 \times 12$ gross code lattice does not support the last two basis metachecks.

\subsection{Translational periods of symmetries}

As illustrated in the gross code example, upon computing a basis of symmetries via the decomposition, we are then able to determine translations which act trivially on each basis element and thus determine a basis of translations which act trivially on the symmetry space, called a coarse-graining of the code family~\cite{Haah_2013,wang2026decoupling2dtranslationinvarianttopological}. Coarse-grainings are used to determine the scales at which 2D codes decouple in toric codes. To determine the periods of a symmetry, we determine the periods of the plane wave factor (e.g. $\omega^{i-j}$) and the local structure factor (e.g. $\binom{i}{2} + ij \omega + \binom{j}{2} \omega^2$) and intersect the corresponding lattices. 

Let us begin with the periods of the plane wave factor, $\alpha^i\beta^j$; if the orbit of $(\alpha,\beta)$ has length $n$ and $\omega^\ell = \alpha$, $\omega^m = \beta$ for integers $\ell$ and $m$ where $\omega$ is a primitive of $\mathbb{K}$, then the periods of the factor are given by the lattice $\{ (i,j) \,\vert\, \ell i + mj \equiv 0 \; \text{mod} \,(2^n-1), \;\mathrm{for}\; i,j \in \mathbb{Z}\}$ (i.e., translations by these lattice vectors leave the factor unchanged). 

Now, let us consider the local structure factor. It has the following form:
\begin{align}
    \sum_k \alpha^{-a_k} \beta^{-b_k} \binom{i}{a_k}\binom{j}{b_k} .
\end{align}
By finding horizontal and vertical periods of the $\binom{i}{a_k}$ and $\binom{j}{b_k}$ terms respectively, we can then take the lowest common multiple of all horizontal and all vertical periods to get a rectangular period of the local structure factor.
The periodicity of $\binom{i}{a_k}$ in $i$ is $2^{\lfloor \log a_k \rfloor + 1}$ (this is a consequence of Lucas' theorem~\cite{lucas_theorem} that $\binom{i}{a_k} = 1$ exactly when $i \land a_k = a_k$ where $i$ and $a_k$ are written as bit-strings), and analogously for ${\scriptstyle\binom{j}{b_k}}$ in $j$. We thus have that the smallest rectangular period of the local structure factor is $\{(2^{\lfloor \log \max_k a_k \rfloor + 1},0),(0,2^{\lfloor \log \max_k b_k \rfloor + 1})\}$, which is a sub-lattice of the lattice of all periods. 

Let us call the lattice of rectangular periods of the local structure factor $L'$, and the lattice of all periods of the structure factor $L$. $L'$ is a sub-lattice of $L$, but the two are not necessarily identical as $L$ may contain diagonal periods. Then, for any $u \in L$ but not in $L'$, there exists $r$ such that $2^r u \in L'$. Let $p_x =(2^{\lfloor \log \max_k a_k \rfloor + 1},0)$ and $p_y = (0,2^{\lfloor \log \max_k b_k \rfloor + 1})$. Then we can check whether $(p_x + p_y) / 2$ is a period. If it is, then we update the basis as $(p_x,(p_x+p_y)/2)$, now covering all $u$ such that $2u \in L'$. We can repeat this; at any stage if our basis is $(u,v)$, we check whether $(u+v)/2$ or $v/2$ are periods and update the basis accordingly (we always keep $p_x$). On the $n$-th iteration we will have covered all $u$ such that $2^n u \in L'$, and upon termination we will have a basis of $L$. Intersecting $L$ with the lattice of periods of the plane wave factor, we then have the lattice of periods of the symmetry space. In Table~\ref{table-periods}, we illustrate the coarse-grainings of various codes. 
\begin{table*}[t]
\small
\setlength{\tabcolsep}{4pt}
\renewcommand{\arraystretch}{1.4} 
\makebox[\textwidth][c]{
\begin{NiceTabular}{|cc|ccc|c|}[cell-space-limits=4pt]
\hline
 \textbf{Code family} & $\bm{A, B}$ & $\bm{(\alpha,\beta)}$ & \textbf{$\mathbb{K}$-syms.} & \textbf{Gen. periods} & \textbf{C.G. periods} \\
\hline
\Block{2-1}{Toric code} & $1+x$ & \Block{2-1}{$(1,1) \in \mathbb{F}_2$} & \Block{2-1}{1} & \Block{2-1}{$(1,0),(0,1)$} & $(1,0)$ \\
& $1+y$ & & & & $(0,1)$ \\
\hline
\Block{2-1}{BB5 code~\cite{Ye_2025}} & $1+x$ & \Block{2-1}{$(1,\omega) \in \mathbb{F}_4$} & \Block{2-1}{$\omega^j$} & \Block{2-1}{$(1,0),(0,3)$} & $(1,0)$ \\
& $1+y + x^2 y^2$ & & & & $(0,3)$\\
\hline
\Block{2-1}{Color code} & $1+x+xy$ & \Block{2-1}{$(\omega,\omega) \in \mathbb{F}_4$} & \Block{2-1}{$\omega^{i+j}$} & \Block{2-1}{$(1,-1),(1,2)$} & $(3,0)$ \\
& $1+y + xy$ & & & & $(0,3)$\\
\hline
\Block{4-1}{Gross code~\cite{Bravyi_2024}} & \Block{2-1}{$1+x+y^3/x$} & $(\omega,\omega) \in \mathbb{F}_4$ & $\omega^{i+j}$ & $(1,-1),(1,2)$ & \Block{2-1}{$(12,0)$} \\
\faintcline{3-5} 
& & \Block{3-1}{$(\omega,\omega^2) \in \mathbb{F}_4$} & $\omega^{i-j}$ & $(1,1),(1,-2)$ & \\
& \Block{2-1}{$1+y+x^3/y$} & & $\omega^{i-j}(i + \omega j)$ & $(2,2),(2,-4)$ & \Block{2-1}{$(0,12)$}\\
& & & $\omega^{i-j} \sum_{k=0}^2 \tbinom{i}{2-k} \tbinom{j}{k} \omega^k$ & $(4,4),(4,-8)$ & \\
\hline
\Block{4-1}{Two-gross~\cite{Bravyi_2024}} & \Block{2-1}{$1+x+y^3/x$} & $(\omega,\omega) \in \mathbb{F}_4$ & $\omega^{i+j}$ & $(1,-1),(1,2)$ & \Block{2-1}{$(12,6)$} \\
\faintcline{3-5} 
& & \Block{2-1}{$(\omega,\omega^2) \in \mathbb{F}_4$} & $\omega^{i-j}$ & $(1,1),(1,-2)$ & \\
& \Block{2-1}{$x^3 + y^2 + y^7$} & & $\omega^{i-j}(i + \omega j)$ & $(2,2),(2,-4)$ & \Block{2-1}{$(30,-48)$}\\
\faintcline{3-5} 
& & ${(\omega,\omega^{19})} \in \mathbb{F}_{64} := \frac{\mathbb{F}_2[\omega]}{( p_0(\omega) )}$ & $\omega^{i + 19j}$ & $(6,3),(7,-7)$ & \\
\hline

\Block{4-1}{BB code~\cite{Bravyi_2024}} & \Block{2-1}{$1+ x^2 + x^7$} & ${(\omega,\omega)} \in \mathbb{F}_4$ & $\omega^{i+j}$ & $(1,-1),(1,2)$ & \Block{2-1}{$(93,0)$}\\
\faintcline{3-5} 
& & $(\omega,\omega^2) \in \mathbb{F}_4$ & $\omega^{i-j}$ & $(1,1),(1,-2)$ & \\
\faintcline{3-5} 
&\Block{2-1}{$1 + y + x^9/y$}  & $(\omega,\omega^{2}) \in \mathbb{F}_{32}:= \frac{\mathbb{F}_2[\omega]}{( p_3(\omega))}$  & $\omega^{i+2j}$ & $(2,-1),(7,12)$ & \Block{2-1}{$(0,93)$}\\
\faintcline{3-5} 
& & $(\omega,\omega^{7}) \in \mathbb{F}_{32}:= \frac{\mathbb{F}_2[\omega]}{( p_3(\omega))}$ & $\omega^{i+7j}$ & $(3,4),(4,-5)$ & \\
\hline

\Block{2-1}{Generalised T.C~\cite{Liang_2025}} & $1 + x + y/x^2$ & $(\omega,\omega^6) \in \mathbb{F}_{16} := \frac{\mathbb{F}_2[\omega]}{( \omega^4 + \omega + 1)}$ & $\omega^{i+6j}$ & $(3,2),(3,-3)$ & $(9,1)$\\
\faintcline{3-5} 
& $1 + y + xy^2$ & $(\omega,\omega^5) \in \mathbb{F}_{8} := \frac{\mathbb{F}_2[\omega]}{(\omega^3 + \omega + 1)}$ & $\omega^{i+5j}$ & $(2,1),(1,-3)$ & $(3,12)$\\
\hline

\Block{2-1}{Generalised T.C~\cite{Liang_2025}} & $1 + x + 1/(xy^3)$ & $(\omega,\omega^{66}) \in \mathbb{F}_{128} := \frac{\mathbb{F}_2[\omega]}{( p_0(\omega))}$ & $\omega^{i+6j}$ & $(5,-2),(11,21)$ & $(127,0)$\\
\faintcline{3-5} 
& $1 + y + 1/y^6$ & $(\omega,\omega^{52}) \in \mathbb{F}_{128}:= \frac{\mathbb{F}_2[\omega]}{(\omega^7 + \omega^4 + 1)}$ & $\omega^{i+5j}$ & $(6,-5),(11,12)$ & $(0,127)$\\
\hline

\Block{4-1}{Generalised T.C~\cite{Liang_2025}} & \Block{2-1}{$1+x+y/x^3$}
& \Block{2-1}{$(\omega,\omega^2) \in \mathbb{F}_4$} & $\omega^{i-j}$ & $(1,1),(1,-2)$ & \Block{2-1}{$(102,-102)$}\\
& & & $\omega^{i-j}(i + \omega j)$ & $(2,2),(2,-4)$ & \\
\faintcline{3-5} 
& \Block{2-1}{$1+y+x/y^3$} & $(\omega,\omega) \in \mathbb{F}_8 := \frac{\mathbb{F}_2[\omega]}{( \omega^3 + \omega^2 + 1)}$ & $\omega^{i+j}$ & $(1,-1),(3,4)$ & \Block{2-1}{$(54,156)$}\\
\faintcline{3-5} 
& & $(\omega,\omega^{16}) \in \mathbb{F}_{256}:= \frac{\mathbb{F}_2[\omega]}{(p_2(\omega))}$ & $\omega^{i+16j}$ & $(1,-16),(16,-1)$ &\\
\hline
\end{NiceTabular}
}
\caption{
    Infinite-lattice symmetries of various 2D code families. We include the defining polynomials $A$ and $B$, one component $(\alpha,\beta)$ from each orbit in the symmetry space decomposition, the $\mathbb{K}$-basis of symmetries $\{b_i\}$ (i.e., the image of $x^iy^j$ under the symmetry map) which produces $\mathbb{F}_2$-symmetries $\{\textrm{Tr}(cb_i) \,\vert\, c \in \mathbb{K}\}$, the generating lattice periods of each symmetry basis element, and the generating periods of the entire symmetry space for coarse-graining (C.G.). We construct larger fields as the univariate polynomial ring $\mathbb{F}_2[x]$ modulus some polynomial, wherein $\omega$ is the image of $x$.  We define the following: $p_0(x) = (\sum_{i=0}^7 \omega^i) - \omega^4$, $p_1(x) = x^6 + x^4 + x^3 + x + 1$, $p_2(x) = x^8 + x^7 + x^5 + x^3 + 1$, and $p_3(x) = x^5 + x^4 + x^2 + x + 1$.
}
\label{table-periods}
\end{table*}

\section{Symmetries on periodic boundaries}\label{section-finite}

As seen in the gross code example of Section~\ref{example-gross-code}, our representation of the infinite-lattice metachecks provides insight into which metachecks are supported on some periodic lattice configuration, as well as the smallest lattice required to support all metachecks. In this section, we approach this systematically with the goal of identifying bases of metachecks on rectangular lattices. In particular, having precomputed an infinite lattice basis of metachecks, we would like to determine the rectangular-lattice basis of metachecks from a boundary configuration with minimal computation. As we will see, there are two parts of this problem: the first is determining which components of the decomposition `survive' on the rectangular lattice (which is to say some nonzero subspace is supported), and the second is determining the surviving subspace of each such component. The idea is that we can think of the boundary ideal $\mathfrak{b} = (x^\ell - 1, y^m - 1)$ as having its own space of `$\mathfrak{b}$-metachecks' $(R/\mathfrak{b})^*$. Our goal is then to compute the intersection of infinite-lattice and $\mathfrak{b}$-metachecks, which will be the space of surviving metachecks.

First, let us explore the structure of $\mathfrak{b}$; given its simple form we expect the decomposition to be correspondingly simple. Let $a$ and $b$ be the largest powers of 2 dividing $\ell$ and $m$. Then given a point $(\alpha,\beta)$ satisfying $\alpha^\ell-1=\beta^m-1=0$, the corresponding component of $\mathfrak{b}$ is $(u^a,v^b)$. This is because we can write $x^\ell - 1 = (x^{\ell/a}-1)^a$; so the multiplicity of roots in $x^\ell - 1$ is given by $a$ (analogously for $b$). So a basis of $\mathbb{K}$-symmetries of $\mathfrak{b}$ is given by $\{\delta_{u^i v^j} \;\vert\; 0 \leq i < a,\, 0 \leq j < b\}$. Let us call this space $\mathcal{B}$.

This suggests the following procedure; for each point $(\alpha,\beta) \in \mathcal{Z}$ satisfying $A=B=0$, we test that the point satisfies the boundary equations $\alpha^\ell - 1 = \beta^m - 1 = 0$. 
If the point satisfies the equations, then we intersect the corresponding component of $\mathbb{K}$-symmetries with $\mathcal{B}$. After repeating for every component/point, we can then convert back to $\mathbb{F}_2$-symmetries to get a decomposition of the symmetries on rectangular boundaries. 

\begin{example}[Gross code symmetries]
    Consider the gross code on $6 \times 12$ rectangular boundaries. $\omega^6-1=\omega^{12}-1=0$, so the $(\omega,\omega)$ component survives. Its basis is $\delta_1$, which is unchanged by intersection with $\mathcal{B} = \mathrm{span}\{\delta_{u^i} \delta_{v^j} \;\vert\; 0 \leq i < 2,\, 0 \leq j < 4\}$. In general, unless a surviving point is fat, its associated space of $\mathbb{K}$-symmetries will be $\{\delta_1\}$, fit in every box, and thus be unchanged on rectangular boundaries. The $(\omega,\omega^2)$ component, however, has a basis element $\delta_{v^2} + \omega^2 \delta_{u^2}$ which pokes outside the box formed by $\mathcal{B}$ (Fig.~\ref{fig:gross_code_box}). So the surviving subspace is spanned by the other two basis elements $\delta_1$ and $\delta_{v} + \omega \delta_u$. Thus, of our infinite-lattice basis of symmetries, all except for $x^i y^j \mapsto \mathrm{Tr}(c \omega^{i-j}\left(\binom{i}{2} + ij \omega + \binom{j}{2} \omega^2\right))$ survive on the $6 \times 12$ lattice.
\end{example}

\begin{figure*}[htbp]
    \centering
    \includegraphics[width=0.48\textwidth]{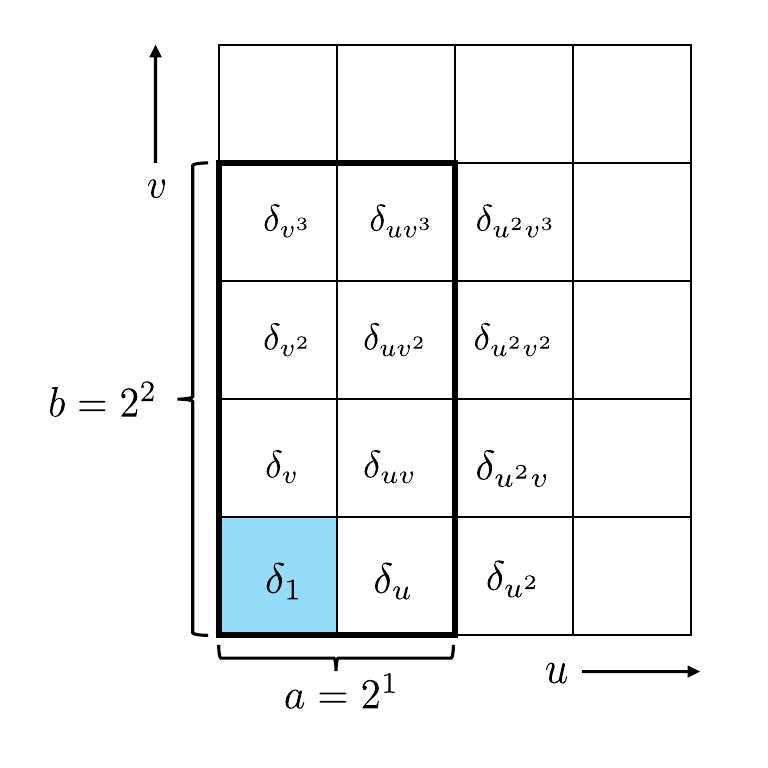}
    \hfill
    \includegraphics[width=0.48\textwidth]{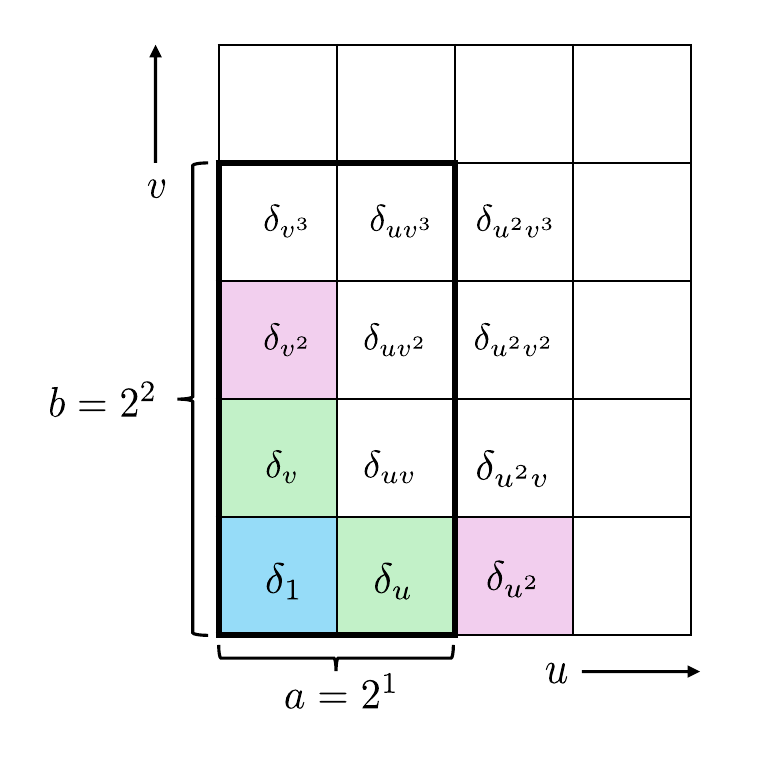}
    \caption{Computing symmetries of the gross code on $6 \times 12$ boundaries. (Left) The $(\omega,\omega)$-component of the gross code $\mathbb{K}$-symmetries is spanned by $\{\delta_1\}$ (blue). This fits in any non-zero box. (Right) The $(\omega,\omega^2)$-component of the gross code $\mathbb{K}$-symmetries is spanned by $\{\delta_1,\delta_v + \omega \delta_u, \delta_{v^2} + \omega^2 \delta_{u^2}\}$ (blue, green, and pink respectively). The pink basis element does not fit inside the box.}
    \label{fig:gross_code_box}
\end{figure*}

In the above example, the basis after intersecting with $\mathcal{B}$ is the subset of the infinite-lattice basis that lies in $\mathcal{B}$. As in Fig.~\ref{fig:adapting_basis}, this is not generally true for an arbitrary basis. We can however adapt the basis of infinite-lattice $\mathbb{K}$-symmetries within each component to make this true. To see this, consider the following ordering of basis elements of $R^*$ $$\delta_1 < \delta_u < \delta_{u^2} \dots < \delta_{v} < \delta_{uv} < \dots$$
As it resembles the $u < v$ monomial order $1 < u < u^2 \dots < v < uv < \dots$, we also call it the $u < v$ ordering. We can similarly conceive a $v < u$ ordering. Our goal is to ensure that the leading term of each basis element of our infinite-lattice $\mathbb{K}$-symmetries under both orderings is distinct; that this is possible is a standard result about bases bi-adapted to two flags (\cite{glaudo2025simultaneousgeneratingsetsflags}, Theorem 2.1) (the flags in our instance are the sequences of nested subspaces induced by either ordering). First, we write our basis of $\mathbb{K}$-symmetries as a matrix with the columns ordered $u < v$. We then convert to row echelon form by Gaussian elimination, so that the leading terms of each basis element under $u < v$ are distinct. We reorder the columns under the $v < u$ ordering and again perform Gaussian elimination to put it in row echelon form. When doing this, however, if we have two rows sharing a leading term, we always pivot off the row whose leading term in the $u < v$ ordering is last. Pivoting in this way will ensure that the leading terms of each row under the $u < v$ ordering remain distinct. 

\begin{figure*}[htbp]
    \centering
    \includegraphics[width=0.48\textwidth]{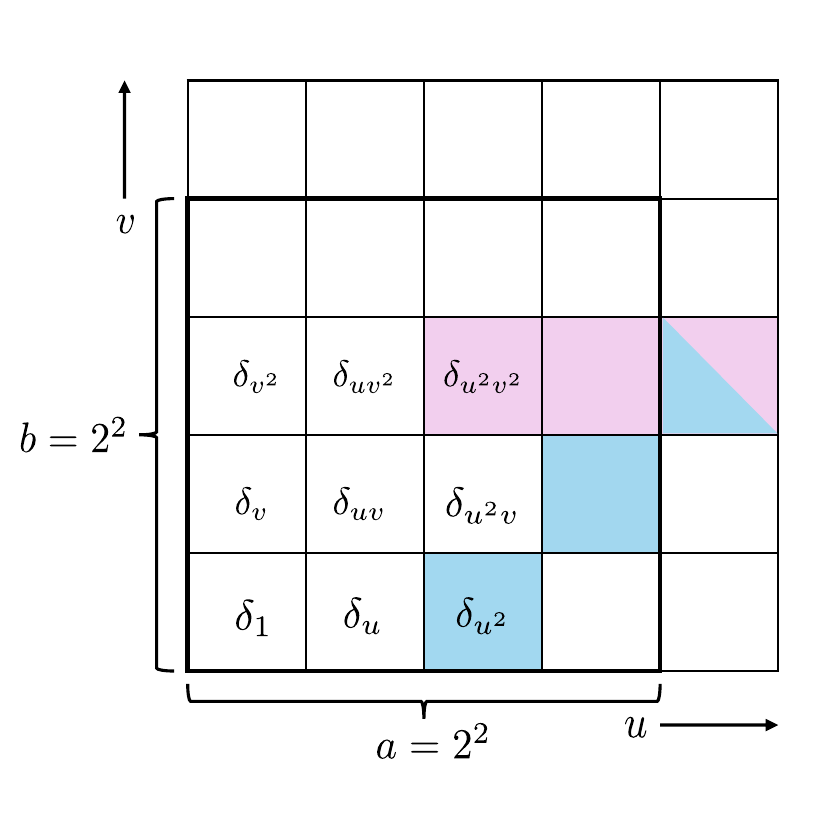}
    \hfill
    \includegraphics[width=0.48\textwidth]{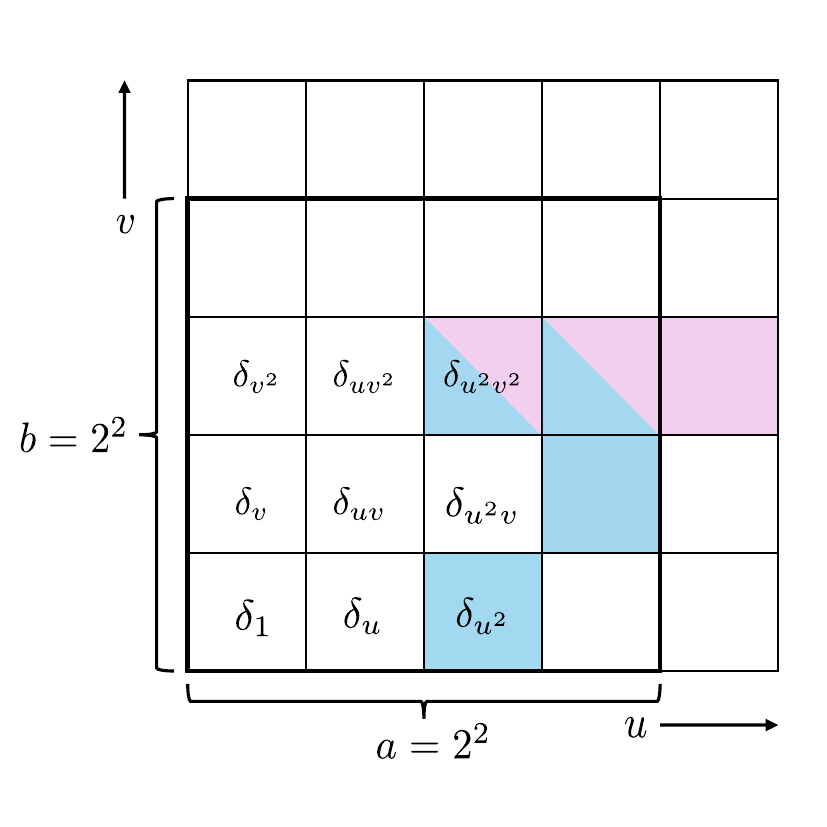}
    \caption{Adapting a basis for rectangular boundaries. (Left) An example of a basis $\{u^2v^2 + u^3 v^2 + u^4 v^2, u^2 + u^3 v + u^4 v^2\}$ where neither basis element (highlighted in pink and blue respectively) survives intersection with $\mathcal{B}$, but their sum does. (Right) An `adapted' basis of the same subspace where upon intersection with $\mathcal{B}$, the blue basis element survives and the pink does not. The blue basis element spans the subspace after intersection.}
    \label{fig:adapting_basis}
\end{figure*}

\begin{table*}[t]
\small
\renewcommand{\arraystretch}{1.4} 
\setlength{\tabcolsep}{4pt}
\makebox[\textwidth][c]{
\begin{NiceTabular}{|c|cc|ll|}[cell-space-limits=4pt]
\hline
 \textbf{Code} & $\bm{A, B}$ & \textbf{Lattice periods} & $\bm{(\alpha,\beta)}$ & $\bm{((R'/J)^*)_{\alpha,\beta}}$ \\
\hline

\Block{2-1}{Toric code $[[2n^2,2,n]]$} & $1+x$ & $(n, 0)$ & \Block{2-1}{$\hl{(1,1)} \in \mathbb{F}_2$ (1)} & \Block{2-1}{$\{\hl{\delta_1}\}$}\\
& $1+y$ & $(0, n)$ & & \\
\hline

\Block{2-1}{BB5 [[30,4,5]]~\cite{Ye_2025} }  & $1+x$ 
    & (5, 0)  & \Block{2-1}{${\hl{(1,\omega)}} \in \mathbb{F}_4$ (2)} & \Block{2-1}{$\{\hl{\delta_1}\}$}\\
    & $1+y + x^2y^2$ & (0, 3) & & \\
\hline

\Block{2-1}{Color code $[[18n^2,4,4n]]$~\cite{eberhardt2024pruningqldpccodesbivariate}} & $1+x+xy$
    & $(3n,0)$ & \Block{2-1}{$\hl{(\omega,\omega)} \in \mathbb{F}_4$ (2)} & \Block{2-1}{$\{\hl{\delta_1}\}$}\\
    & $1+y+xy$ & $(0,3n)$ & & \\
\hline

\Block{2-1}{Gross code [[144,12,12]]~\cite{Bravyi_2024}} & $1+x+y^3/x$
    & (6, 0) & $\hl{(\omega,\omega)} \in \mathbb{F}_4$ (2) & $\{\hl{\delta_1}\}$\\
    & $1 + y + x^3/y$ & (0, 12) & $\hl{(\omega,\omega^2)} \in \mathbb{F}_4$ (4)& $\{\hl{\delta_1}, \hl{\delta_v + \omega \delta_u}, \delta_{v^2} + \omega^2 \delta_{u^2}\}$\\
\hline
\Block{2-1}{$12 \times 12$ Gross [[288,16,12]]~\cite{cruzbenito2026evolutionarydiscoverybivariatebicycle}} & $1+x+y^3/x$
    & (12, 0) & $\hl{(\omega,\omega)} \in \mathbb{F}_4$ (2) & $\{\hl{\delta_1}\}$ \\
    & $1 + y + x^3/y$ & (0, 12) & $\hl{(\omega,\omega^2)} \in \mathbb{F}_4$ (6)& $\{\hl{\delta_1}, \hl{\delta_v + \omega \delta_u}, \hl{\delta_{v^2} + \omega^2 \delta_{u^2}}\}$\\
\hline

\Block{2-1}{$6 \times 9$ Gross [[108,8,10]]~\cite{Bravyi_2024}} & $1+x+y^3/x$
    & (6, 0) & $\hl{(\omega,\omega)} \in \mathbb{F}_4$ (2) & $\{\hl{\delta_1}\}$ \\
    & $1 + y + x^3/y$ & (0, 9) & $\hl{(\omega,\omega^2)} \in \mathbb{F}_4$ (2)& $\{\hl{\delta_1}, {\delta_v + \omega \delta_u}, {\delta_{v^2} + \omega^2 \delta_{u^2}}\}$\\
\hline

\Block{6-1}{Two-gross [[288,12,18]]~\cite{Bravyi_2024}}
    & \Block{3-1}{$1 + x + y^3/x$} & \Block{3-1}{(12, 0)}
    & \Block[l]{2-1}{$\hl{(\omega,\omega)} \in \mathbb{F}_4$ (2)}
    & \Block[l]{2-1}{$\{\hl{\delta_1}\}$} \\[-0.7\normalbaselineskip]
    & & & & \\[-0.7\normalbaselineskip]
    & & & \Block[l]{2-1}{$\hl{(\omega,\omega^2)} \in \mathbb{F}_4$ (4)}
    & \Block[l]{2-1}{$\{\hl{\delta_1}, \hl{\delta_v + \omega \delta_u}\}$} \\[-0.7\normalbaselineskip]
    & \Block{3-1}{$x^3 + y^2 + y^7$} & \Block{3-1}{(0, 12)} & & \\[-0.7\normalbaselineskip]
    & & & \Block[l]{2-1}{$(\omega,\omega^{19}) \in \mathbb{F}_{64} := \frac{\mathbb{F}_2[\omega]}{(p_0(\omega))}$}
    & \Block[l]{2-1}{$\{\delta_1\}$} \\
    & & & & \\
\hline

\Block{4-1}{BB code [[90,8,10]]~\cite{Bravyi_2024}} & \Block{2-1}{$1+x^2 + x^7$}
    & \Block{2-1}{(15, 0)} & $\hl{(\omega,\omega)} \in \mathbb{F}_4$ (2) & $\{\hl{\delta_1}\}$ \\
    & & & $\hl{(\omega,\omega^2)} \in \mathbb{F}_4$ (2)& $\{\hl{\delta_1}\}$\\
    & \Block{2-1}{$1 + y + x^9/y$}& \Block{2-1}{(0, 3)} & $(\omega,\omega^{2}) \in \mathbb{F}_{32}  := \frac{\mathbb{F}_2[\omega]}{( p_3(\omega))}$ & $\{{\delta_1}\}$\\
    & & & $(\omega,\omega^{7}) \in \mathbb{F}_{32}  := \frac{\mathbb{F}_2[\omega]}{( p_3(\omega))}$ & $\{{\delta_1}\}$\\
\hline

\Block{2-1}{Generalised T.C [[120,8,12]]~\cite{Liang_2025}} & $1+x+y/x^2$
    & (0, 10) & $\hl{(\omega,\omega^6)} \in \mathbb{F}_{16} := \frac{\mathbb{F}_2[\omega]}{( \omega^4 + \omega + 1)}$ (4)& $\{\hl{\delta_1}\}$\\
    & $1+y+xy^2$ & (6, 4) & $(\omega,\omega^5) \in \mathbb{F}_8 := \frac{\mathbb{F}_2[\omega]}{(\omega^3 + \omega + 1)}$ & $\{\delta_1\}$\\
\hline

\Block{2-1}{Generalised T.C [[254,14,16]]~\cite{Liang_2025}} & $1+x+1/(xy^3)$
    & (0, 127) & $\hl{(\omega,\omega^{66})} \in \mathbb{F}_{128} := \frac{\mathbb{F}_2[\omega]}{( p_0(\omega))}$ (7)& $\{\hl{\delta_1}\}$\\
    & $1+y+1/y^6$ & (1, 25) & $(\omega,\omega^{52}) \in \mathbb{F}_{128} := \frac{\mathbb{F}_2[\omega]}{(\omega^7 + \omega^4 + 1)}$ & $\{\delta_1\}$\\
\hline

\Block{6-1}{Generalised T.C [[294,10,20]]~\cite{Liang_2025}}
    & \Block{3-1}{$1+x+y/x^3$} & \Block{3-1}{(0, 21)}
    & \Block[l]{2-1}{$\hl{(\omega,\omega^2)} \in \mathbb{F}_4$ (2)}
    & \Block[l]{2-1}{$\{\hl{\delta_1}, \delta_v + \omega \delta_u\}$} \\[-0.7\normalbaselineskip]
    & & & & \\[-0.7\normalbaselineskip]
    & & & \Block[l]{2-1}{$\hl{(\omega,\omega)} \in \mathbb{F}_{8} := \frac{\mathbb{F}_2[\omega]}{(\omega^3 + \omega^2 + 1)}$ (3)}
    & \Block[l]{2-1}{$\{\hl{\delta_1}\}$} \\[-0.7\normalbaselineskip]
    & \Block{3-1}{$1+y+x/y^3$} & \Block{3-1}{(7, 7)} & & \\[-0.7\normalbaselineskip]
    & & & \Block[l]{2-1}{$(\omega,\omega^{16}) \in \mathbb{F}_{256} := \frac{\mathbb{F}_2[\omega]}{(p_2(\omega))}$}
    & \Block[l]{2-1}{$\{\delta_1\}$} \\
    & & & & \\
\hline
\end{NiceTabular}
}
\caption{
    A table of the infinite-lattice symmetries of various 2D codes with the ones surviving compactification to the finite lattice highlighted in red. The $\mathbb{K}$-symmetries are given instead of the $\mathbb{F}_2$-symmetries for compactness. In all cases, the surviving elements form a basis of finite-lattice symmetries. In parentheses we put the dimension of the component associated with each orbit after compactification. We also define $p_0(x) = (\sum_{i=0}^7 x^i) - x^4$, $p_1(x) = x^6 + x^4 + x^3 + x + 1$, $p_2(x) = x^8 + x^7 + x^5 + x^3 + 1$, and $p_3(x) = x^5 + x^4 + x^2 + x + 1$.
}
\label{table-examples}
\end{table*}

What we need to precompute then is a bi-adapted basis of $\mathbb{K}$-symmetries on the infinite lattice (to be precise, we do not need the whole basis as we ignore redundant components in the same orbit) as well as the corresponding basis of infinite-lattice $\mathbb{F}_2$-symmetries. We can then take any rectangular boundary configuration $(\ell,m)$ and first, for each point $(\alpha,\beta)$ test it against the boundary equations. For points satisfying the equations, we then select the sub-basis that lies in $\mathcal{B}$. The corresponding $\mathbb{F}_2$-symmetries then form a basis on the rectangular lattice. In Table.~\ref{table-examples}, we explicitly describe the symmetries of various examples from the literature. Although the procedure outlined above is specific to rectangular boundaries, we have in Table.~\ref{table-examples} codes on non-rectangular or twisted boundaries; in many cases, the components are either all trivial (in the sense that there are no fat points) or relatively simple, and thus we can easily determine how the symmetry space varies with different choices of twisted boundary. We can also determine the logical dimension or number of logical qubits of these codes, which is double the dimension of the symmetry space. Suppose for each $(\alpha_i,\beta_i) \in \mathbb{F}_{2^{q_i}}$ (with one point per orbit, and the field being the smallest for that orbit), $n_i$ $\mathbb{F}_{2^{q_i}}$-symmetry basis elements survive. Then the logical dimension is $\sum_i q_i n_i$. We can thus explicitly describe how the logical dimension varies with rectangular lattice size. 

\begin{example}[Counting logical qubits]
    Consider the family of codes with gross polynomials on lattice dimensions $2^a \ell \times 2^b m$, with $\ell,m$ odd, the logical dimension is 0 unless both $\ell$ and $m$ are multiples of 3. If $\mathrm{min}(a,b)=0$, then the logical dimension is 8, with each orbit contributing 4. If $\mathrm{min}(a,b)=1$, then the logical dimension is 12, with the $(\omega,\omega^2)$ orbit contributing 8, and if $\mathrm{min}(a,b)=2$, then the logical dimension is 16. More generally, for twisted boundaries with corners $(0,0),\,(2^a \ell),\, (2^b m, 2^c n),\, (2^a \ell + 2^b m, 2^c n)$, we have four logical qubits from the $(\omega,\omega)$ point if $3 \,\vert\, \ell, m+n$. We additionally have $4 \times (\min\,\{a,b,c,2\}+1)$ logical qubits from the $(\omega,\omega^2)$ point if $3 \,\vert\, \ell, m-n$.
\end{example}

\section{Conclusion}

In this work, we have described a procedure to compute the infinite-lattice materialised symmetries/global metachecks of BB codes by decomposing the symmetry space into components closed under translation. Furthermore, we gain explicit expressions for the symmetries resembling plane waves that reveal their translational structure. This builds on prior methods based on \grobner bases \cite{Liang_2025} that return $\mathbb{F}_2$-vector space bases that describe the additive but not translational properties of symmetries. These bases can be used to determine coarse-grainings to decouple bivariate bicycle codes into toric codes. We show how to easily compute a minimal rectangular coarse-graining, i.e.,\ a smallest rectangular boundary configuration that preserves all infinite-lattice symmetries. We also compute a basis of all coarse-grainings, including twisted coarse-grainings, although this is more computationally involved.

We show how to use the infinite-lattice symmetries to efficiently compute the symmetries of the associated family of codes on different rectangular periodic lattices: first, appropriately adapt the bases of each component, then by determining whether each basis element `survives', determine a sub-basis of the infinite-lattice symmetries that describes the finite-lattice symmetries. A byproduct of this is that we can compute the number of logical qubits on finite lattices, and can thus given a family of 2D TI codes describe how its logical dimension varies as a function of the lattice dimensions.

We note that our method for finding infinite-lattice symmetries also generalises to trivariate-tricycle code families~\cite{jacob2026singleshotdecodingfaulttolerantgates,Menon_2026} with growing distances. Here, we have a duality between the module $R/(A,B,C)$ of (Z) point excitations and (X) global symmetries, and decomposing the module of pointlike excitations gives us an analogous decomposition of the symmetry space. For example, the $[[72,6,6]]$ code in~\cite{jacob2026singleshotdecodingfaulttolerantgates} has symmetries given by $x^i y^j z^k \mapsto \mathrm{Tr} (\omega^{j+c})$ for $c \in \mathbb{Z}_3$ and $\omega \in \mathbb{F}_4$ (we thus see that as its symmetry space is 2-dimensional, it must have at most $2 \times 3 = 6$ logical qubits). In 3-dimensions, we cannot however adapt the basis to cheaply determine finite-lattice bases of symmetries; this is because we can at most adapt a basis to 2 flags (in the 2D case, these came from our orderings $u>v$ and $v>u$). In many cases, such as the above example, it should be fairly straightforward to identify how the basis changes on periodic boundaries.

Our method does not apply in full generality to twisted boundaries; the twisted boundary ideals do not appear to admit an elegant decomposition. If the infinite lattice components do not include fat points or the twisted boundaries have odd width and height, then restricting to any boundary is straightforward and amounts to checking which points $(\alpha,\beta)$ satisfy the boundary condition equations.
Even when this is not the case, one can often determine how the symmetry space changes across different twisted boundaries by examining the explicit expressions for the infinite symmetry basis, as demonstrated with the gross code and its variants, given these expressions are sufficiently uncomplicated.

We have considered only 2D codes following a 2-block CSS form.  However, we expect that our work should generalise to non-CSS codes. It is known that for a 2D code on a square lattice with growing distance, stabiliser generators can be chosen so that the code has exclusively global metachecks (\cite{Haah_2013} Lemma 7.1). In such codes, we expect anyon-symmetry duality to hold, and thus we can decompose the anyon module into local pieces to get a symmetry decomposition. The anyon module in this general case is obtained by quotienting the module of excitations by the excitation/syndrome map that takes Pauli operators to their respective syndromes. 

Our results will likely find direct application in developing and improving matching-based decoders~\cite{sahay2026matchingdecoderbivariatebicycle,tan2026generalizedmatchingdecoders2d} for 2D codes, and also understanding what makes certain BB codes amenable to such decoders. Matching-based decoders have provable thresholds and therefore avoid having error floors on sufficiently large lattices~\cite{tan2026generalizedmatchingdecoders2d}, in contrast to heuristic decoders where such guarantees may not exist. In practice, performance of matching decoders can suffer due to the decrease in effective distance incurred by decoupling the code into symmetries. Heuristically, by including as many low periodicity symmetries as possible in the selected basis, the performance could be improved~\cite{tan2026generalizedmatchingdecoders2d}.  As each component in the symmetry-space decomposition corresponds to an irreducible periodicity (in the sense that we cannot write a symmetry with this periodicity as a sum of lower periodicity symmetries), the resulting basis satisfies these demands.
Should we want to design BB codes as matching-friendly as possible, we might want the components of the symmetry space to be as low in periodicity as possible.
It may thus be interesting to select ideals $I$ with low-periodicity decompositions, and test matching-based decoders on BB codes where $(A,B) = I$.

Another future direction would be to port this framework to understand open-boundary BB codes \cite{steffan2025tilecodeshighefficiencyquantum,liang2025planarquantumlowdensityparitycheck,wang2026algebrabivariatebicyclesurfacecodes} and design a larger variety of boundary configurations beyond rough-smooth boundaries. If we include all boundary gauge operators, the symmetries are consistent with those on the infinite lattice. Excluding boundary gauge operators to gap the boundary then `breaks' certain symmetries, so the surviving symmetries classify the boundary. 

It would also be interesting to explore the global metachecks/symmetries of fracton codes such as Haah's cubic code \cite{Haah_2011}, where the ideal decomposition is significantly more complex and we encounter 1D curves in addition to points (see Fig.~\ref{fig:primary-decomp}); the component associated with each 1D curve is an infinite-dimensional vector space. Understanding metachecks in such codes could pave the way to matching-based decoders of fracton codes, as well as improving our understanding of how the code properties vary with different boundary configurations.

\section*{Author Contributions}

AL developed the procedure described and obtained all results in this manuscript. CKM led the project. All authors contributed towards the writing of the manuscript and the direction of the project. OpenAI ChatGPT 5.4/5.6 and Claude Opus 4.8 was used to identify relevant mathematical literature and tools. The authors take full responsibility for the content of the manuscript. 

\section*{Acknowledgements}

We thank Benjamin Brown, Kaavya Sahay, and Dominic Williamson for discussions.  This work is supported by the Defence Science and Technologies Group (DSTG) and Advanced Strategic Capabilities Accelerator (ASCA) through its Emerging and Disruptive Technologies (EDT) Program, and the Intelligence Advanced Research Projects Activity (IARPA) through the Entangled Logical Qubits program Cooperative Agreement Number W911NF-23-2-0223.

\bibliography{main}
\appendix

\section{Proofs}

\subsection{Proof of Proposition.~\ref{prop-am-duality}}
\label{proof1}

\begin{proof}
    We will show that $\ker \tilde{d^0} \cong (R/(A,B))^*$; the second statement follows analogously. The isomorphism $\varphi: \ker \tilde{d^0} \to (R/(A,B))^*$ is defined as follows: $\varphi(f)([p]) = f(p)$ where $[p]$ is the class of $p$ in $R/(A,B)$. First, we show $\varphi$ is well-defined; that $f(p)=f(q)$ if $p-q \in (A,B)$. Let $p-q = rA + sB$. We have that $f(p-q) = f(rA+sB) = \tilde{A}f(r) = \tilde{B}f(s) = 0$, and thus $f(p)=f(q)$.
    It is easily checked that $\varphi (f + g) = \varphi (f) + \varphi (g)$ and $\varphi (m f) = m \varphi (f)$ for a monomial $m$. It follows that $\varphi$ is a homomorphism. Let us similarly define $\varphi^{-1}: (R/(A,B))^* \to \ker \tilde{d^0}$, with $\varphi^{-1} (u)(p) = u([p])$. This is readily checked to also be a homomorphism, and it can also be verified that it is the inverse of $\varphi$.
\end{proof}

\subsection{Proof of Proposition.~\ref{prop-dual-decomp}}
\label{proof2}
\begin{proof}
    We want to show that $\oplus_{(\alpha,\beta) \in Z} \left(R'/J_{(\alpha,\beta)}\right)^* \cong \left(\oplus_{(\alpha,\beta) \in Z} R'/J_{(\alpha,\beta)}\right)^*$.
    As $\mathrm{Hom}_{K}( -, K)$ sends direct sums to direct sums in our finitely-many summand case, we have that the two are isomorphic as $K$-vector spaces; the isomorphism takes a tuple of functions $(f_1,\dots,f_n)$ to the function $f$ with $f(a_1,\dots,a_n) = \sum_i f_i(a_i)$.
    It is clear that this respects $R'$-multiplication; if $(f_1, \dots f_n)$ is mapped to $f$, then $r \cdot (f_1, \dots f_n) = (rf_1, \dots, rf_n)$ gets mapped to $g$ with $g(a_1, \dots,a_n) = \sum_i (rf_i)(a_i) = \sum_i f_i(r^\dagger a_i) = f(r^\dagger (a_1, \dots a_n)) = rf(a_1,\dots,a_n)$.
\end{proof}
\bigskip
\subsection{Proof of Theorem.~\ref{prop-trace-decomp}}\label{proof3}

First, we show some intermediate but elementary results.
We will use a different notation from the theorem statement. Applying a CRT decomposition to $(R/I)$ gives us $R/I = \oplus_i R/Q_i$; we can then write the CRT decomposition of $R'/J$ as $\oplus_{ij} R/Q_{ij}$, where $Q_{i1},Q_{i2},\dots$ form an orbit (these are what we would write as $J_{(\alpha,\beta)},J_{(\alpha^2,\beta^2)},\dots$) such that $Q_{i} R' = \cap_i Q_{ij}$. Our $S_{(\alpha,\beta)}$'s will turn out to be the components $R/Q_i$.

\begin{lemma}\label{contraction}
    $Q_i = Q_{ij} \cap R$
\end{lemma}
\begin{proof}
    Suppose we order the components so $\sqrt{Q_{ij}} = (x-\alpha^{2^i}, y-\beta^{2^i})$. Consider the Frobenius automorphism $\sigma$ that squares the coefficients of polynomials. Then, as $Q_i R' = \cap_i Q_{ij}$, and noting $\sigma {Q_i R'} = Q_i R'$, we see that $\cap_i \sigma(Q_{ij})$ gives an alternate decomposition of $Q_i R'$ with the primes corresponding to each $Q_{ij}$ being cycled. However, as the local decomposition of $R'/Q_i R'$ is unique, $\sigma$ must permute components $Q_{ij}$. It follows that $Q_{ij} \cap R = Q_{ij'} \cap R$ as $R$ consists fixed points under $\sigma$. Thus $Q_{ij} \cap R = \cap_{j'} Q_{ij'} \cap R = Q_i R' \cap R = Q_i$.
\end{proof}

\begin{lemma} \label{field-trace-nondegen}
    Given a $K$-vector space $V$, and a $\mathbb{F}_2$-linear map $g: V \to \mathbb{F}_2$, there exists a $K$-linear map $f: V \to K$ such that $g(v) = \mathrm{Tr}(f(v))$.
\end{lemma}
\begin{proof}
    We use the fact that all $\mathbb{F}_2$-linear maps from $K$ to $\mathbb{F}_2$ can be written as $x \mapsto \mathrm{Tr}(ax)$. Given an $\mathbb{F}_2$-linear map $g: V \to \mathbb{F}_2$ for a $K$-vector space $V$, we can consider the map $c \mapsto g (cv)$ for $c \in K$ and some $v \in V$, which we must be able to write as $\mathrm{Tr}(ac)$ for some $a$. Let us denote this $a$ as $f(v)$, thus defining $f: V \to K$. We can show $f$ is $K$-linear as follows: taking $d \in K$, we have that for all $c \in K$, $\mathrm{Tr}(cf(dv)) = g(cdv) = \mathrm{Tr}(cd f(v))$. As the trace pairing $(a,b) \mapsto \mathrm{Tr}(ab)$ is non-degenerate, we have that $f(dv) = df(v)$. Similarly, $\mathrm{Tr}(cf(v+w)) = g(cv + cw) = \mathrm{Tr}(c(f(v)+f(w)))$, and thus $f(v+w)=f(v)+f(w)$ by non-degeneracy. So we see that $f$ is $K$-linear. As $g(v) = \mathrm{Tr}(f(v))$, we are done.
\end{proof}
Now, we can show the main result.

\begin{proof}
    Consider the map $s \mapsto \mathrm{Tr}(s) \circ \iota$. It is readily seen that this takes a component such as $R'/Q_{i1}$ into $R/Q_i$, but what we need to show is that it surjects. Consider the map $x / Q_i \mapsto \iota(x) /Q_{i1}$ from $R/Q_i$ to $R'/Q_{i1}$. It follows from Lemma.~\ref{contraction} that this map is injective and thus takes a basis of $R/Q_i$ to a linearly independent set in $R'/Q_{i1}$. We can thus extend some map $s: R/Q_i \to \mathbb{F}_2$ to a map $s': R/Q_{i1} \to \mathbb{F}_2$ such that $s = s' \circ \iota$. Applying Lemma.~\ref{field-trace-nondegen}, we have that for any $s \in (R/Q_i)^*$, we have $s' \in (R/Q_{i1})^*$ such that $s = \mathrm{Tr}(s') \circ \iota$.
\end{proof}

\end{document}